\pdfoutput=1
\documentclass[11pt]{article}
\usepackage[letterpaper,margin=1in]{geometry}

\usepackage[utf8]{inputenc} 
\usepackage[T1]{fontenc}    
\usepackage{url}            
\usepackage{booktabs}       
\usepackage{amsfonts}       
\usepackage{nicefrac}       
\usepackage{microtype}      
\usepackage{xcolor}         

\usepackage{nicefrac,bm,bbm}
\usepackage{amsmath,amsthm,color,colortbl,amssymb}

\usepackage{hyperref}

\usepackage{caption}
\usepackage{natbib}
\renewcommand{\bibsection}{\section*{References}}
\usepackage{wrapfig}
\usepackage{comment}

\usepackage{graphicx}
\usepackage{multirow}
\usepackage{xspace}
\usepackage{mathtools}
\usepackage{euscript}
\usepackage{thm-restate}
\usepackage{tikz}
\usepackage{subcaption}
\usepackage{tablefootnote,tabularx}
\usepackage[T1]{fontenc}    

\usepackage{pifont}

\usepackage{euscript}
\usepackage{mleftright}

\usepackage[capitalize]{cleveref}

\crefalias{AlgoLine}{line}%
\crefname{algocf}{Algorithm}{Algorithms}
\usepackage[noend]{algorithm}

\usepackage{cancel}
\usepackage{newpxtext}
\usepackage{newpxmath}
\definecolor{niceRed}{RGB}{190,38,38}
\definecolor{Red2}{RGB}{219, 50, 54}
\definecolor{mgreen}{RGB}{160, 200, 140}
\definecolor{blueGrotto}{RGB}{5,157,192}
\definecolor{limeGreen}{HTML}{81B622}
\definecolor{myellow}{rgb}{0.88,0.61,0.14}
\definecolor{darkGreen}{HTML}{2E8B57}
\definecolor{navyBlueP}{HTML}{03468F}
\definecolor{Sepia}{HTML}{7F462C}
\definecolor{red2}{HTML}{1F462C}
\definecolor{orange2}{HTML}{FF8000}
\definecolor{mgray}{HTML}{ABB3B8}
\definecolor{myPurple}{RGB}{175,0,124}
\definecolor{mypurple2}{rgb}{0.8,0.62,1}
\definecolor{royalBlue}{HTML}{057DCD}
\definecolor{mpink}{HTML}{FC6C85}
\theoremstyle{plain}
\newtheorem{theorem}{Theorem}[section]

\newtheorem{lemma}[theorem]{Lemma}

\theoremstyle{definition}

\theoremstyle{remark}
\newtheorem{remark}[theorem]{Remark}

\hypersetup{
	colorlinks = true,
	linkcolor = niceRed,
	citecolor = royalBlue,
	linktocpage = true,
	urlcolor = mpink
}

\usepackage{amsmath,amsfonts,bm}
\usepackage{amssymb}
\usepackage{mathtools}
\usepackage{amsthm}

\def\eqref#1{equation~\ref{#1}}

\def\1{\bm{1}}

\DeclareMathAlphabet{\mathsfit}{\encodingdefault}{\sfdefault}{m}{sl}
\SetMathAlphabet{\mathsfit}{bold}{\encodingdefault}{\sfdefault}{bx}{n}

\newcommand{\Rset}{\mathbb{R}^{m}_{+}}
\newcommand{\Actions}{\mathcal{A}}

\DeclareMathOperator*{\argmax}{arg\,max}

\newcommand{\sfP}{\mathsf{P}}
\newcommand{\sfA}{\mathsf{A}}

\usepackage{tikz}
\usetikzlibrary{arrows.meta, positioning}
\usepackage{pgfplots}
\usepgfplotslibrary{fillbetween}

\usepackage{url}
\allowdisplaybreaks

\usepackage{multirow}
\usepackage{tikz}
\usepackage[utf8]{inputenc} 
\usepackage[T1]{fontenc}    
\usepackage{hyperref}       
\usepackage{url}            
\usepackage{booktabs}       
\usepackage{amsfonts}       
\usepackage{nicefrac}       
\usepackage{microtype}      
\usepackage{xcolor}         
\usepackage{amsmath}
\usepackage{amsthm,enumerate,amssymb}
\usepackage{thm-restate}
\usepackage{thmtools} 
\usepackage{mathtools}
\usepackage{wrapfig}
\usepackage{amsmath}
\usepackage{pifont}

\usepackage{comment}
\usepackage{graphics}
\usepackage{graphicx}
\usepackage{appendix}
\usepackage{algorithm}
\usepackage{enumitem}
\usepackage[noend]{algpseudocode}
\algdef{SE}[DOWHILE]{Do}{doWhile}{\algorithmicdo}[1]{\algorithmicwhile\ #1}%
\usepackage{bbm}
\algnewcommand{\IfThen}[2]{
	\State \algorithmicif\ #1\ \algorithmicthen\ #2}

\usepackage{tikzscale} 

\usepackage[short, nocomma]{optidef} 

\newcommand{\A}{\mathcal{A}}

\usepackage[textsize=tiny]{todonotes}

\algblockdefx[Execute]{Execute}{EndExecute}[1]{\textbf{until} #1 \textbf{run}}{\textbf{end execute}}
\makeatletter
\ifthenelse{\equal{\ALG@noend}{t}}%
{\algtext*{EndExecute}}
{}%
\makeatother

\algblockdefx[IfInline]{IfInline}{EndIfInline}[1]{\textbf{if} #1 \textbf{then} }{\textbf{end if}}
\makeatletter
\ifthenelse{\equal{\ALG@noend}{t}}%
{\algtext*{EndIfInline}}
{}%
\makeatother

\definecolor{mygreen}{rgb}{0.0, 0.5, 0.0}
\definecolor{myorange}{rgb}{0.55, 0.62, 1}

\title{Contract Design Under Approximate Best Responses}

\author{}
\renewcommand{\eqref}[1]{(\ref{#1})}

\author{
	{Francesco Bacchiocchi}$^\dag$ \and
	{Jiarui Gan}$^\ddag$\and
	{Matteo Castiglioni}$^\dag$ \and
	{Alberto Marchesi}$^\dag$ \and
	{Nicola Gatti}$^\dag$ \\\\
	 $^\dag$ Politecnico di Milano\\
	$^\ddag$  University of Oxford\\\\
	\texttt{francesco.bacchiocchi@polimi.it, jiarui.gan@cs.ox.ac.uk}\\
	\texttt{ \{matteo.castiglioni, alberto.marchesi, nicola.gatti\}@polimi.it}
}

\begin{document}
\maketitle
\begin{abstract}
	Principal-agent problems model scenarios where a principal incentivizes an agent to take costly, unobservable actions through the provision of payments.
	Such problems are ubiquitous in several real-world applications, ranging from blockchain to the delegation of machine learning tasks.
	In this paper, we initiate the study of hidden-action principal-agent problems under \emph{approximate best responses}, in which the agent may select \emph{any} action that is not too much suboptimal given the principal's payment scheme (a.k.a.~\emph{contract}).
	Our main result is a polynomial-time algorithm to compute an optimal contract under approximate best responses.
	This positive result is perhaps surprising, since, in Stackelberg games, computing an optimal commitment under approximate best responses is computationally intractable.
	We also investigate the learnability of contracts under approximate best responses, by providing a no-regret learning algorithm for a natural application scenario where the principal has \emph{no} prior knowledge about the environment.
\end{abstract}
\section{Introduction}\label{sec:introduction}

In \emph{hidden-action principal-agent} problems, a principal tries to steer the behavior of a self-interested agent toward favorable outcomes.
The agent has to take a costly action that stochastically determines an outcome resulting in a reward for the principal.
The main challenge is that the agent's action is \emph{hidden} to the principal, who can only observe the realized outcome.
Thus, the principal influences the agent's behavior by \emph{committing to a contract},
%
which is an outcome-dependent payment scheme whose aim is to induce the agent to take a high-cost action leading to high principal's rewards.
The principal's goal is to design an \emph{optimal} contract, namely one maximizing their utility, \emph{i.e.}, rewards minus payments.

Nowadays, principal-agent problems find application in a terrific number of real-world settings, such as, \emph{e.g.}, crowdsourcing~\citep{ho2015adaptive}, online labor platforms~\citep{kaynar2023estimating}, blockchain~\citep{cong2019blockchain}, delegation of machine learning tasks~\citep{cai2015optimum}, and pay-for-performance healthcare~\citep{bastani2016analysis,bastani2019evidence}.
Moreover, algorithmic contract design is playing a crucial role in today's world, which increasingly relies on AI agents to perform complex tasks (see, \emph{e.g.},~\citep{hadfield2019incomplete,saig2024incentivizing}).

Previous works on algorithmic contract design assume that the agent always plays a \emph{best-response action} to the principal's contract.
However, if the agent actually responds (even slightly) suboptimally to the principal, then the principal's utility may substantially deteriorate.
This may be the case in most of the real-world applications of interest, for several different reasons.
For instance, the principal  may \emph{not} perfectly know agent's features and account for the wrong best response, the agent may \emph{not} be powerful enough to compute an (exact) best-response action, or they may inaccurately interpret the principal's contract.

In this paper, we initiate the study of hidden-action principal-agent problems \emph{under approximate agent's best responses}.
Specifically, we consider settings in which the agent may take actions that are up to $\delta \in (0,1)$ suboptimal under the principal's contract.
%
We do \emph{not} make any assumption on the specific $\delta$-best response selected by the agent, but we allow them to take \emph{any} of such actions.
Thus, we take a \emph{worst-case approach} and consider the problem of designing contracts under the assumption that the agent selects the worst $\delta$-best response for the principal.
Contracts designed in such a way are said to be \emph{robust}, as they guard the principal against any possible (conceivable) suboptimal behavior of the agent.

\subsection{Results and Techniques}

In this paper, we provide an extensive treatment of the computational and learnability aspects of the design of robust contracts.
The results presented in the paper are organized into three main parts as summarized below.

\paragraph{The Price of Robustness}
In the first part of the paper, as a preliminary analysis, we provide a characterization of the maximum utility that the principal can achieve by means of robust contracts. 
Specifically, we provide upper and lower bounds on this utility, as a function of a parameter $\delta \in (0,1)$ quantifying the agent's best response suboptimality.
These bounds give insights on the \emph{price (in terms of utility) that the principal incurs for being robust}, as the parameter $\delta$ can be seen as a measure of the robustness level of the principal's contracts.
Interestingly, our results show that, differently from what happens in general Stackelberg games (see, \emph{e.g.},~\citep{gan2024robust}), upper/lower bounds do \emph{not} depend on the inducibility gap $\Delta>0$ characterizing the problem instance.
In order to derive the bounds, we prove that it is possible to convert any non-robust contract into a robust one by properly moving its payments in the direction of principal’s rewards, and that, in an optimal robust contract, which provides the principal with a utility greater than zero, the agent’s utility should be at least $\delta$.

\paragraph{Computing Robust Contracts}
The second part of the paper addresses the problem of computing a robust contract that is optimal for the principal.
Our main contribution is a polynomial-time algorithm for this problem.
%
%
This is perhaps surprising, since analogous problems in Stackelberg and Bayesian persuasion settings are known to be computationally intractable~\citep{gan2024robust,yang2024computational},
despite having more amenable solution spaces ($\Delta_m$ in these problems vs. $\mathbb{R}_+^m$ in contract design).
%
%
At a high level, our algorithm cleverly exploits a particular structure that we discover in the robustness constraints (which ensure that the agent plays the worst approximate best response for the principal).
%
%
%
Intuitively, we show that, when the agent's best response and their (worst) approximate best response are fixed to two arbitrarily-selected actions, the problem of computing a utility-maximizing robust contract can be formulated by means of a union of $n+1$ different \emph{linear programs} (LPs).
Therefore, by taking the maximum over all these LPs, one can compute a utility-maximizing robust contract once the agent’s best response and their (worst) approximate best response are fixed.
By further iterating this procedure for all the possible choices of the action pair, we then determine an optimal robust contract.

\paragraph{Learning Robust Contracts} 
In the third and last part of the paper, we investigate the learnability of robust contracts.
Specifically, we study an online learning framework similar to the one analyzed by~\citet{Zhu2023Sample}, in which the features of agent's actions, \emph{i.e.}, costs and probabilities over outcomes, depend on an agent's \emph{type} that is sampled at each round from some (fixed) unknown probability distribution.
%
%
%
At each round, after committing to some contract, the principal only observes the outcome (and its associated reward) realized as an effect of the approximate best response played by the agent.
Our main result within this online learning framework is the design of an algorithm that achieves sublinear regret with respect to always playing an optimal robust contract.
Additionally, we show that when the parameter $\delta \in (0,1)$ measuring the suboptimality of agent's actions is sufficiently small with respect to the time horizon, our algorithm also achieves sublinear regret with respect to always playing an optimal \emph{non-robust} contract. 
Our approach presents some advantages compared to the state-of-the-art proposed by~\citet{Zhu2023Sample} for the non-robust version of the problem, while achieving similar regret guarantees.
Indeed, our algorithm employs a simpler discretization of the set of contracts employed by the principal, which does \emph{not} require that the principal knows its rewards beforehand.
This is made possible by the fact that we rely on a novel ``continuity'' argument for the principal’s expected utility (see Lemma~\ref{lem:epsilonconvert}),  which is different from the one originally proposed by~\citet{Zhu2023Sample}.

\subsection{Related Works}
\paragraph{Robustness to approximate best responses} 
Two works that are closely related to ours are~\citep{gan2024robust,yang2024computational}, which consider \emph{robustness} notions analogous to ours, though in different settings.
Specifically,~\citet{gan2024robust} initiate this research line, by studying th problem of computing robust leader's commitments in \emph{Stackelberg games} where the follower plays an approximate best response.
%
%
They show that it is \textsf{NP}-hard to approximate an optimal robust commitment of the leader and, in accordance to this hardness result, they provide a quasi-polynomial-time approximation scheme (QPTAS).
\citet{yang2024computational} extend the study initiated by~\citet{gan2024robust} to \emph{Bayesian persuasion}, with the goal of computing robust signaling schemes under approximate best responses of the receiver.
Similarly to~\citet{gan2024robust}, they show that computing an approximately-optimal robust signaling scheme is \textsf{NP}-hard and provide a QPTAS.
In sharp contrast, we show that, in hidden-action principal-agent problems, an optimal robust contract can be computed efficiently.
%
\paragraph{Contract design} 
Contract theory has been extensively studied in economics~\cite{holmstrom1991multitask}. 
However, the interest in its computational aspects is more recent. 
\citep{dutting2019simple} analyzes linear contracts, proving approximation bounds. \citep{castiglioni2022bayesian,guruganesh2021contracts} show that computing optimal contracts in Bayesian settings is intractable. 
\citep{alon2023bayesian} studies Bayesian linear contracts, proving near-optimality under sufficient uncertainty. 
Other works explore combinatorial principal-agent problems~\citep{babaioff2006combinatorial,babaioff2009free} and multi-agent extensions of hidden-action problems~\citep{castiglioni2023multi,duetting2024multiagent}.

\paragraph{Other forms of robustness in contract design} 
Our work is also related to other research lines addressing different concepts of \emph{robustness} in contract design.
%
%
For instance, \citet{carroll2015robustness} studies settings where the principal only knows a superset of agents' actions, while~\citet{dutting2019simple} introduce a different notion of uncertainty in which the principal has partial knowledge of the distributions over outcomes associated with agent's actions.
Both these works show that linear contracts are a sufficient class of contracts to determine the \emph{min-max robust} optimal contract.
Notice that these two frameworks differ from ours, as within our framework, when the robustness parameter $\delta \in (0,1)$ is arbitrarily small, the problem becomes very close to the classical version of the hidden-action principal-agent problem, in which it is known that linear contracts are \emph{not} generally optimal (see, \emph{e.g.},~\citep{dutting2019simple}).
Recently, \citet{bernasconi2024regret} study settings where uncertainty lies in the costs of agent's actions.
In this framework, the principal only knows a set containing the true cost vectors, and computing an optimal min-max robust contract is \textnormal{APX}-hard.
\paragraph{Learning in principal-agent problems} 
Our work is also related to online learning problems in hidden-action principal-agent problems.
\citet{Zhu2023Sample} study general hidden-action principal-agent problem instances in which the principal faces multiple agent's types.
They show that it is possible to design an algorithm that achieves a regret bound of the order of $ \widetilde{\mathcal{O}}(\sqrt{m} \cdot T^{1 - 1/(2m+1)})$ when the principal selects contracts from the hypercube $[0,1]^m$, where $m$ is the number of outcomes.
In our work, we show that it is possible to design an algorithm achieving similar regret guarantees even when the different agent's types select approximate best responses.
Our algorithm presents some advantages compared to the one proposed by~\citet{Zhu2023Sample}. 
Specifically, our approach employs a simpler discretization of the hypercube---used during the execution of the algorithm---compared to~\cite{Zhu2023Sample}, and it does \emph{not} require prior knowledge of the principal's rewards.
%
Furthermore,~\citet{Zhu2023Sample} provide an (almost-matching) lower bound of $\Omega(T^{1-1/(m+2)})$, which holds even with a single agent's type.
Some recent works have introduced additional hypothesis to overcome this negative result (see, \emph{e.g.},~\cite{bacchiocchilearning,chen2024bounded}).
\section{Preliminaries}\label{sec:preliminaries}

In this section, we first introduce the classical hidden-action principal-agent problems (Section~\ref{sec:prelim_model_classical}), and then the variation studied in this paper, in which the agent plays an approximate best response (Section~\ref{sec:prelim_model_robust}).

\subsection{Hidden-Action Principal-Agent Problems}\label{sec:prelim_model_classical}

An instance of \emph{hidden-action principal-agent problem} is characterized by a tuple $\left(\mathcal{A},\Omega,F,r,c\right)$, where $\mathcal{A}$ is a finite set of $n \coloneqq |\mathcal{A}|$ actions available to the agent, $\Omega$ is a finite set of $m \coloneqq |\Omega|$ possible outcomes, $F \in [0,1]^{m \times n}$ is a matrix representing the effects of agent's action, $r \in [0,1]^m$ is a reward vector for the principal, and $c \in [0,1]^n$ is a vector of agent's costs.
Each agent's action $a \in \mathcal{A}$ determines a probability distribution over outcomes, encoded by a column $F_a \in \Delta_{\Omega}$ of the matrix $F$, and it results in a cost for the agent, encoded by a component $ c_a \in [0, 1]$ of vector $c$.\footnote{We denote by $\Delta_X$ the set of all probability distributions over the set $X$.
	Given $n \in \mathbb{N}_{>0}$, we write $[n]\coloneqq \{1,\dots, n\}$.}
%
	%
	%
%
We denote by $F_{a,\omega} \in [0,1]$ the probability with which action $a$ results in outcome $\omega \in \Omega$, as prescribed by $F_a$.
Thus, it must be the case that $\sum_{\omega \in \Omega}F_{a,\omega}=1$ for all $a \in \mathcal{A}$.
Each outcome $\omega \in \Omega$ is associated with a reward for the principal, encoded by a component $r_\omega \in [0,1]$ of the vector $r$.
Thus, whenever the agent selects an action $a \in \mathcal{A}$, the principal's expected
reward can be computed as $R_a\coloneqq F_{a}\cdot r$.\footnote{We denote by 
	$x\cdot y$ the dot product of two vectors $x,y \in \mathbb{R}^d$.}

The principal commits to an outcome-dependent payment scheme called \emph{contract}, which is a vector $p \in \mathbb{R}^{m}_{+}$ defining a payment $p_\omega \geq 0$ from the principal to the agent for every outcome $\omega \in \Omega$.\footnote{As customary in contract theory~\citep{carroll2015robustness}, we assume that the agent has \emph{limited liability}, meaning that the payments can only be from the principal to the agent, and \emph{not viceversa}.}
Given a contract $p \in \mathbb{R}^{m}_{+}$, the agent plays a \emph{best-response} action that is: {(i)} \emph{incentive compatible} (IC), which means that it maximizes their expected utility; and {(ii)} \emph{individually rational} (IR), meaning that it has non-negative expected utility.
We assume w.l.o.g. that there always exists at least one \emph{opt-out} action $a \in \mathcal{A}$ with null cost, \emph{i.e.}, $c_a = 0$, and for which the principal's expected reward is equal to zero, \emph{i.e.}, $ F_{a} \cdot  r = 0$.
Since the agent's utility is at least zero when playing the \emph{opt-out} action, any IC action is also IR, allowing us to focus on incentive compatibility~only.

Whenever the principal commits to a contract $p \in \mathbb{R}^{m}_{+}$ and the agent responds by playing an action $a \in \mathcal{A}$,
the agent's and the principal's expected utilities are, respectively,
\begin{align*}
	u^\sfA(p,a)\coloneqq F_{a} \cdot p - c_a, 
	\;\text{ and }\;
	u^\sfP(p,a)\coloneqq F_{a} \cdot (r-p).
\end{align*}
%
%
The set $A(p) \subseteq \mathcal{A}$ of agent's best responses in a contract $p \in \mathbb{R}^{m}_{+}$ is defined as follows:
\begin{equation*}
	A(p) \coloneqq \arg\max_{a \in \mathcal{A}} \left \{ F_a \cdot p - c_a \right\}.
\end{equation*}
In classical (non-robust) hidden-action principal-agent problems, the agent breaks ties in favor of the principal when having multiple best responses available (see, \emph{e.g.}, \citep{dutting2019simple}).
We denote by $a(p) \in A(p)$ the action played by the agent in a given contract $p \in \mathbb{R}^{m}_{+}$.
This is an action $a \in A(p)$ that maximizes the principal's utility $F_{a} \cdot \left( r - p \right)$.
Formally, $a(p) \in \argmax_{a \in A(p)} F_{a} \cdot \left( r - p \right) $.
%
%
Then, the goal of the principal is to design a contract $p \in \mathbb{R}^{m}_{+}$ that maximizes their expected utility $u^\sfP(p,a(p))$.
We say that a contract $p^\star \in\mathbb{R}^{m}_{+}$ is a (non-robust) \emph{optimal} contract if it holds $p^\star \in \arg\max_{p \in \mathbb{R}^{m}_{+}} u^{\textnormal{P}}(p,a(p))$.
In the following, we define the principal's utility in an optimal (non-robust) contract as $\textnormal{OPT} := \max_{p \in \mathbb{R}^{m}_{+}} u^{\textnormal{P}}(p, a(p))$, while we let the value of the social welfare be $\textnormal{SW} := \max_{a \in \mathcal{A}} \left\{ F_a \cdot r - c_a \right\}$.
%

\subsection{Robust Contracts and Approximate Best Responses}\label{sec:prelim_model_robust}

In this paper, we study a variation of the classical hidden-action principal-agent problem, where the agent plays an action that is an \emph{approximate} best response.
Given $\delta \in (0,1)$, we define the set $A^\delta(p) \subseteq {A}$ of agent's $\delta$-\emph{best responses} in a given contract $p \in \mathbb{R}^{m}_{+}$ as follows:
\begin{equation*}
	A^{\delta}(p) \hspace{-.5mm} \coloneqq \hspace{-.5mm} \left \{ a \in \mathcal{A} \mid F_{a} \hspace{-.5mm}\cdot\hspace{-.5mm} p - c_a \hspace{-.5mm}>\hspace{-.5mm} \max_{a' \in \mathcal{A}} \left\{ F_{a'}\hspace{-.5mm} \cdot\hspace{-.5mm} p - c_{a'} \right\}\hspace{-.5mm} - \hspace{-.5mm}\delta \right\}.
\end{equation*}
%
%
We adopt an adversarial robust approach, in the sense that, whenever the principal commits to a contract $p\in\mathbb{R}^{m}_{+}$, the agent selects a $\delta$-best response that minimizes principal's expected utility, namely an action $a^\delta(p) \in 	A^{\delta}(p)$ such that $a^\delta(p) \in \arg\min_{a \in A^\delta(p)}F_{a}\cdot (r-p)$.
We refer to the utility of a $\delta$-\emph{robust} contract $p$ as $u^\sfP(p,a^\delta(p))$.

Given $\delta \in (0,1)$, the principal's goal is to design an \emph{optimal} $\delta$-robust contract $p^\star \in \Rset$, which is formally defined as:
\begin{equation}\label{eq:delta_opt_def}
	p^\star  \in \argmax_{p \in \Rset} \min_{a \in A^\delta(p)} u^\sfP(p,a) = \argmax_{p \in \Rset} \Psi(p),
\end{equation}
where $\Psi(p) \coloneqq \min_{a \in A^\delta(p)} u^\sfP(p, a)$ denotes the principal's expected utility in a $\delta$-robust contract $p\in \Rset$.
%
%
Notice that analogous ``$\delta$-robust solution concepts'' have been already introduced in similar settings, namely Stackelberg games~\citep{gan2024robust} and Bayesian persuasion~\citep{yang2024computational}.
Our $\delta$-robust contracts are their analogous in the context of hidden-action principal-agent problems.

In the following, given $\delta \in (0,1)$, we denote the expected utility of the principal in an optimal $\delta$-robust contract $p^\star $ as $\textnormal{OPT}(\delta) \coloneqq u^\sfP(p^\star, a^\delta(p^\star)).$
We remark that $\textnormal{OPT}(\delta)$ is always well defined, as an optimal $\delta$-robust contract, according to the definition in \Cref{eq:delta_opt_def}, always exists.
Intuitively, this is due to the strict inequality in the definition of the set $A^\delta(p)$, as observed by~\citet{gan2024robust} for Stackelberg games.
Thus, in order to prove the existence of an optimal $\delta$-robust contract, it is possible to employ the same argument used to prove
Proposition~1 in~\citep{gan2024robust}.
%

\section{The Price of Robustness}\label{sec:characterizations}

We start by providing a characterization of how the value $\textnormal{OPT}(\delta)$ of an optimal $\delta$-robust contract varies as a function of the parameter $\delta \in (0,1)$, which controls the suboptimality of the agent's best response.
The goal of our analysis is to quantify how much is the \emph{price (in terms of utility) that the principal incurs for being robust} to agent's approximate best responses. 
Indeed, the parameter $\delta$ can be interpreted as a measure of the robustness level of the principal's contract, with higher $\delta$ values indicating higher levels of robustness.

We first establish upper and lower bounds that identify a suitable region in which the values of $\textnormal{OPT}(\delta)$ are contained.
Such bounds only depend on the parameter $\delta \in (0,1)$, the value of an optimal non-robust contract $\textnormal{OPT}$, and the social welfare $\textnormal{SW}$ achievable with non-robust contracts.
%
\begin{restatable}[Upper and lower bounds]{proposition}{RealtionNonRobust}\label{pro:properties}
	Given an instance of hidden-action principal-agent problem:
	\begin{enumerate}
		\item For every $\delta \in (0,1)$, it holds:
		\begin{equation*}
			\textnormal{OPT}(\delta) \ge \textnormal{OPT}_{\textnormal{LB}}(\delta) \coloneqq \textnormal{OPT} - 2 \sqrt{\delta} + \delta .
		\end{equation*}
		\item For every $\delta \in (0,1)$, it holds:
		\begin{equation*}
			\textnormal{OPT}(\delta)\le \textnormal{OPT}_{\textnormal{UB}}(\delta) \coloneqq \max\left\{0, \, \textnormal{SW} - \delta \right\}. 
		\end{equation*}
	\end{enumerate}
\end{restatable}
In order to prove the first point in Proposition~\ref{pro:properties}, we show that it is always possible to convert a non-robust contract into a robust one by suitably moving it towards the direction of the principal's reward vector.
In order to prove the second point of Proposition~\ref{pro:properties}, we observe that, if an optimal $\delta$-robust contract $p^\star$ provides the principal with an expected utility strictly larger than zero, then the opt-out action does \emph{not} belong to the set of $\delta$-best responses for such a contract.
Therefore, in an optimal $\delta$-robust contract that provides the principal with an expected utility strictly larger than zero, the agent's utility is at least $\delta > 0$, and thus the principal's expected utility is at most $\textnormal{SW} - \delta$.
%
%

The following proposition shows that the upper and lower bounds in Proposition~\ref{pro:properties} are tight. Formally:
\begin{restatable}{proposition}{RealtionNonRobustTwo}\label{pro:tightness}
	Given $\delta \in (0,1)$, for every integer $n > n(\delta)$, there is an instance of hidden-action principal-agent problem (parametrized by $\delta$) with $2n + 1$ actions, where:
	\begin{equation*}
		\textnormal{OPT}(\delta) -\textnormal{OPT}_{\textnormal{LB}}(\delta) \le \mathcal{O}\left(\frac{1}{n}\right).
	\end{equation*}
	%
	%
	Furthermore, for any $\delta\in(0,1)$, there exists an instance of hidden-action principal-agent problem in which it holds: $$\textnormal{OPT}(\delta)= \textnormal{OPT}_{\textnormal{UB}}(\delta).$$
\end{restatable}

%
%
%
Proposition~\ref{pro:tightness} shows that our upper bound $\textnormal{OPT}_{\textnormal{UB}}(\delta)$ is (strictly) tight, while our lower bound $\textnormal{OPT}_{\textnormal{LB}}(\delta)$ is tight up to an additive term that linearly goes to zero as the number of agent's actions $n$ increases.

Finally, thanks to the previous propositions, we can show the following property about the value of an optimal $\delta$-robust contracts as a function of $\delta \in (0,1)$.
\begin{restatable}{proposition}{RealtionNonRobustThree}\label{pro:properties2}
	The function $(0,1) \ni \delta \mapsto \textnormal{OPT}(\delta)$ is continuous and non-increasing in $\delta$.
	Moreover, $\lim_{\delta \to 0^+} \textnormal{OPT}(\delta)=\textnormal{OPT}$ and $\lim_{\delta \to 1^-} \textnormal{OPT}(\delta)=0$.
	%
	%
\end{restatable}
Figure~\ref{fig:jumps} shows an example of the region in which the value of $\textnormal{OPT}(\delta)$ is bounded, as defined by the lower and upper bounds in Propositions~\ref{pro:properties}~and~\ref{pro:properties2} as functions of $\delta$.
%
%
%

\begin{figure}[!t]
	\centering
	\resizebox{0.55\linewidth}{!}{\begin{tikzpicture}[scale=0.9,transform shape]
	\begin{axis}[
		domain=0:1.03, 
		samples=100, 
		axis lines=middle,
		xlabel={\small$\hspace{3mm}\delta$},
		xlabel style={below right},
		ylabel={\small$\textnormal{OPT}(\delta)$},
		ylabel style={left},
		ymin=0, ymax=1.2,
		xmin=0, xmax=1.1,
		xtick={0,1},
		ytick={0,0.7,0.9,3},
		yticklabels={0, \small $\textnormal{OPT} = 0.7$, \small $\textnormal{SW} =0.9$},
		width=7.5cm,
		height=5.35cm,
		legend pos=south west
		]
		\addplot [
		name path=A,
		thick,
		blue,
		domain=0:1,
		restrict y to domain=0:0.8
		]{min(0.9 - x, 0.7)};
		
		\addplot [
		name path=B,
		thick,
		blue,
		]{0.7 - 2*sqrt(x) + x};
		
		\addplot [
		fill=blue, 
		opacity=0.1
		] fill between[
		of=A and B,
		];
		
		\addplot [
		dashed,
		gray
		] coordinates {(0, 0.7) (1, 0.7)};
		
		\addplot [
		dashed,
		gray,
		domain=0:1
		]{0.9 - x};
		
	\end{axis}
\end{tikzpicture}}
	\caption{The blue area corresponds to the region in which $\textnormal{OPT}(\delta)$ is bounded as a function of $\delta \in (0,1)$ in an instance of hidden-action principal-agent problem with $\textnormal{SW}=0.9$ and $\textnormal{OPT}=0.7$. 
	}
	\label{fig:jumps}
\end{figure}
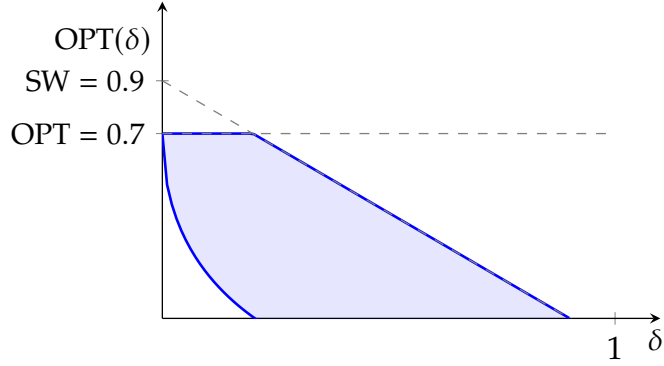

\paragraph{Comparison with Stackelberg games}
The characterization results derived in this section exhibit some crucial differences compared to analogous results derived for Stackelberg games by~\citet{gan2024robust}.
Indeed, in such settings, \citet{gan2024robust}~show that the value of an optimal $\delta$-robust commitment crucially depends on a parameter $\Delta > 0$ that represents the \emph{inducibility gap} of the problem instance.
Intuitively, the inducibility gap encodes how easy it is for the leader to induce the follower to play \emph{any} action; see~\citep{gan2024robust} for a formal definition.
Specifically, in Stackelberg games, the leader's expected utility $\textnormal{OPT}(\delta)$ in an optimal $\delta$-robust commitment is a Lipschitz function lower bounded by $\textnormal{OPT} - \delta / \Delta$ if $\delta < \Delta$, whereas $\textnormal{OPT}(\delta)$ may \emph{not} be even a continuous function if $\delta > \Delta$.
In contrast, in hidden-action principal-agent problems, the value of an optimal $\delta$-robust contract is a continuous function with respect to $\delta \in (0,1)$, regardless of the inducibility gap of the instance.
Furthermore, the value of $\textnormal{OPT}(\delta)$ is either zero or it is upper bounded by $\textnormal{SW} - \delta$, showing that for large values of $\delta$, the maximum utility that the principal can achieve may be particularly small.
Intuitively, this is because the principal must provide the agent with a large expected payment to induce them to take desirable actions rather than the opt-out one.
This upper bound does \emph{not} generally hold in Stackelberg games in which the principal may achieve a large utility even for large values of $\delta$.

\section{Computing Optimal Robust Contracts}\label{sec:computational_problem}

Now, we present our main result: a polynomial-time algorithm to compute an optimal $\delta$-robust contract.
%
%

\subsection{Characterizing an Optimal Contract}

We begin by presenting an optimization problem that characterizes an optimal $\delta$-robust contract. 
Consider an arbitrary optimal $\delta$-robust contract $p^\star$, as well as two arbitrary actions $a^\star \in A(p^\star)$ and $a^\delta \in \arg\min_{a \in A^\delta(p^\star)} u^\sfP(p^\star, a)$.
By fixing $a^\star$ and $a^\delta$, we show that the following optimization problem, over the variable $p \in \mathbb{R}_+^m$, characterizes an optimal $\delta$-robust contract $p^\star$.
\begin{align}
	\label{eq:opt}
	\max_{p\in  \mathbb{R}_+^m}
	\quad  
	u^\sfP(p, a^\delta)
\end{align}
subject to the following \emph{disjunctive} constraints, which must hold for every agent's action $a \in \mathcal{A}$:
\begin{align}
	\hspace{-2mm}\Big( u^\sfA(p, a) \le u^\sfA(p, a^\star) -   \delta  \Big) \vee \Big(u^\sfP(p, a) \ge u^\sfP(p, a^\delta)\Big).\tag{\ref*{eq:opt}a}\label{eq:opt-1}
\end{align}


Intuitively, \Cref{eq:opt-1} requires that each action $a$ is either \emph{not} a $\delta$-best response (first inequality) or no worse than $a^\delta$ for the principal (second inequality). This ensures that the objective function $u^\sfP(p,a^\delta)$ captures the principal's utility in a $\delta$-robust contract $p$.
More formally, we establish the following lemma (recall that $\Psi(p) = \min_{a \in A^\delta(p)} u^\sfP(p, a)$ denotes the principal's utility in a $\delta$-robust contract $p$).


	


	

\begin{lemma}
	\label{lmm:Psi-opt}
	Every optimal solution $p \in \mathbb{R}_+^m$ to Problem~\eqref{eq:opt} is an optimal $\delta$-robust contract, i.e., $\Psi(p) = \Psi(p^\star)$.
\end{lemma}

\begin{proof}
	First, observe that $p^\star$ is a feasible solution to Problem~\eqref{eq:opt}.
	Indeed, if an action $a$ is {\em not} a $\delta$-best response to $p^\star$, then by definition this means $u^\sfA(p^\star, a) \le u^\sfA(p^\star, a^\star) - \delta$; otherwise, it must be that $u^\sfP(p^\star, a) \ge u^\sfP(p^\star, a^\delta)$ as $a^\delta$ is by definition the worst $\delta$-best response in $p^\star$.
	As a result, \Cref{eq:opt-1} holds for $p^\star$ for every $a \in \mathcal{A}$.
	
	Furthermore, notice that, according to the definition of $a^\delta$:
	\[
	\Psi(p^\star) = \min_{a \in A_\delta(p^\star)} u^\sfP(p^\star, a) = u^\sfP(p^\star, a^\delta).
	\]
	
	In order to complete the proof, it is sufficient to show that $u^\sfP(p, a^\delta) \le \Psi(p)$ for every feasible solution $p\in \mathbb{R}_+^m$ to Problem~\eqref{eq:opt}.
	Indeed, if the condition above holds, for any arbitrary optimal solution $p'\in \mathbb{R}_+^m$ to Problem~\eqref{eq:opt}, we have:
	\begin{align*}
		\Psi(p^\star)
		= u^\sfP(p^\star, a^\delta) 
		\le u^\sfP(p', a^\delta)
		\le \Psi(p') 
		\le \Psi(p^\star),
	\end{align*}
	where $u^\sfP(p^\star, a^\delta) \le u^\sfP(p', a^\delta)$ since $p^\star$ is a feasible solution to Problem~\eqref{eq:opt} and $p'$ is optimal,
	while $\Psi(p') \le \Psi(p^\star)$ since $p^\star$ is an optimal $\delta$-robust contract by definition.
	Then, it must be the case that $\Psi(p') = \Psi(p^\star)$.
	
	Now, to complete the proof, consider any feasible solution $p\in \mathbb{R}_+^m$ to Problem~\eqref{eq:opt}. We show that $u^\sfP(p, a^\delta) \le \Psi(p)$.
	%
	Pick an arbitrary best-response action $a' \in A(p)$ and consider any $a \in A^\delta(p)$.
	By definition,
	\begin{align*}
		u^\sfA(p, a) > u^\sfA(p, a') - \delta = \max_{a \in \Actions} u^\sfA(p, a) - \delta 
		\ge u^\sfA(p, a^\star) - \delta.
	\end{align*}
	Hence, for \Cref{eq:opt-1} to hold for action $a$, it must be that $u^\sfP(p, a) \ge u^\sfP(p, a^\delta)$.
	Since the choice of $a$ is arbitrary, this holds for every $a \in A^\delta(p)$.
	Consequently,
	\[
	\Psi(p) = \min_{a \in A^\delta(p)} u^\sfP(p, a) \ge u^\sfP(p, a^\delta).
	\qedhere
	\]
\end{proof}

The above lemma implies that we can effectively ``guess'' $a^\star$ and $a^\delta$, fixing these actions in Problem~\eqref{eq:opt} and solving the optimization problem to obtain $p^\star$ (or possibly a different, but still optimal, $\delta$-robust contract).
Since there are only $O(n^2)$ possible combinations of the values of $a^\star$ and $a^\delta$, the approach is efficient as long as Problem~\eqref{eq:opt} can be solved efficiently.
A correct guess yields a contract $p$ such that $\Psi(p) = \Psi(p^\star)$, whereas $\Psi(p) \le \Psi(p^\star)$ for incorrect guesses.
Thus, by comparing the $\Psi$ values, we can identify a correct guess and a corresponding optimal $\delta$-robust contract.

It remains to show how to efficiently solve Problem~\eqref{eq:opt}.

\subsection{Solving Problem~\eqref{eq:opt}}

Problem~\eqref{eq:opt} does \emph{not} directly admit any efficient solution algorithm due to the non-convex constraint in~\Cref{eq:opt-1}.
To deal with this issue, we rewrite \Cref{eq:opt-1} as follows:
\begin{align}
	\Big( F_a \cdot p \le c_a + u^\sfA(p, a^\star) - \delta \Big) \;\vee  \;  \Big( F_a \cdot p \le F_a \cdot r - u^\sfP(p, a^\delta) \Big),
	\label{eq:opt-1-re}
\end{align}
by expanding the utilities as $u^\sfA(p, a) = F_a \cdot p - c_a$ and $u^\sfP(p, a) = F_a \cdot (r - p)$ and rearranging the terms.

Hence, \Cref{eq:opt-1} is satisfied for action $a \in \mathcal{A}$ if and only if $F_a \cdot p$ is smaller than the maximum of the right-hand sides of the two inequalities in~\Cref{eq:opt-1-re}.
The constraint effectively reduces to the second inequality for all $p \in \Rset$ such that 
\begin{align}
	\label{eq:right-sides}
	c_a + u^\sfA(p, a^\star) - \delta \le F_a \cdot r - u^\sfP(p, a^\delta),
\end{align} 
while for the other $p\in \Rset$, it reduces to the first inequality.

Consequently, we can partition the contract space based on the satisfiability of \Cref{eq:right-sides}, considering all $a \in \mathcal{A}$.
Within each subspace in the partition, only one inequality in \Cref{eq:opt-1-re} is active for every $a \in \mathcal{A}$. So, effectively, \Cref{eq:opt-1-re} reduces to a linear constraint for each action $a$, and the optimization problem to an LP.
It then suffices to solve an LP for every subspace, each generating an optimal contract within its corresponding subspace.
Among these contracts, the one providing the highest utility is an optimal solution to Problem~\eqref{eq:opt}.

Now, if the linear inequalities in \Cref{eq:right-sides} (one for each action $a \in \mathcal{A}$) were $n$ arbitrary inequalities, the above partition may consist of exponentially many subspaces, making the approach inefficient. 
Fortunately, the partition is much more well-structured, since the hyperplanes corresponding to the inequalities are {\em parallel}, as the coefficients of $p$ in \Cref{eq:right-sides} are invariant with respect to $a$.
As a result, they partition the space into only $O(n)$ subspaces.

Next, we show how to exploit the above observation to solve Problem~\eqref{eq:opt}, by solving $O(n)$ suitable subproblems instead.

\subsection{Formulating the Subproblems}

Let us first rearrange \Cref{eq:right-sides} as follows:
\begin{align}
	\label{eq:right-sides-re}
	u^\sfA(p, a^\star) + u^\sfP(p, a^\delta) - \delta \le \nu_a \coloneqq F_a \cdot r - c_a, 
\end{align} 
where $\nu_a$ is exactly the social welfare generated by action $a $ (which is independent of the specific contract adopted).

Then, we can re-order agent's actions $a_1, \dots, a_n$ in such a way that $\nu_{a_1} \le \nu_{a_2} \le \dots \le \nu_{a_n}$. For simplicity, we write $\nu_j = \nu_{a_j}$, and we let $\nu_0 = -\infty$ and $\nu_{n+1} = +\infty$.
Then, the following lemma is straightforward.

\begin{lemma}
	\label{lmm:right-sides}
	For every contract $p \in \mathbb{R}^{m}_{+}$, if it holds that $\nu_{j-1} \leq u^\sfA(p, a^\star) + u^\sfP(p, a^\delta) - \delta \leq \nu_{j}$, then:
	%
	\begin{itemize}
		\item $F_a \cdot p \le F_a \cdot r - u^\sfP(p, a^\delta) \Longleftrightarrow$ \Cref{eq:opt-1-re} holds for all actions $a  \in \{a_\ell \mid \ell \le j - 1\}$; and 
		\item $F_a \cdot p \le c_a + u^\sfA(p, a^\star) - \delta \Longleftrightarrow$ \Cref{eq:opt-1-re} holds for all actions $a  \in \{a_\ell \mid j \le \ell \}$.
	\end{itemize}
\end{lemma}

In other words, the condition in the lemma defines a suitable subspace of contracts $\mathcal{P}_j$ for each $j \in \{1,\dots,n+1\}$. 
The following LP solves for a $p \in \Rset$ that is optimal within $\mathcal{P}_j$. 
\begin{align}
	\label{eq:lp-nu-j}
	\max_{p \in \mathbb{R}_+^m}
	\quad
	u^\sfP(p, a^\delta)
	\tag{\ref*{eq:lp-nu-j-const}}
\end{align}
subject to the following constraints:
\begin{subequations}
	\label{eq:lp-nu-j-const}
	\begin{align}
		&
		\nu_{j-1} \le u^\sfA(p, a^\star) + u^\sfP(p, a^\delta) - \delta \le  \nu_j \label{eq:lp-nu-j-1} \\
		&
		F_a {\cdot} \, p \le F_a \cdot r - u^\sfP(p, a^\delta)
		\quad \forall a \in \{a_\ell \mid  \ell \le j{-}1\}  \label{eq:lp-nu-j-2} \\
		&
		F_a {\cdot} \, p \le c_a + u^\sfA(p, a^\star) - \delta
		\quad\quad \forall a \in \{a_\ell \mid j \le \ell \}. \label{eq:lp-nu-j-3}
	\end{align}
\end{subequations}
Since $\bigcup_{j=1}^{n+1} \mathcal{P}_j = \mathbb{R}^m_+$, solving the LP in Problem~\eqref{eq:lp-nu-j-const} for all $j \in \{1,\dots,n+1\}$ and picking the best among the obtained solutions gives an optimal solution to Problem~\eqref{eq:opt}.

\begin{remark}
	The left-hand side of \Cref{eq:right-sides-re} is roughly (up to a $\delta$ difference) the social welfare of an optimal contract.
	Thus, \Cref{lmm:right-sides} can be interpreted as follows.
	For low-social-welfare actions $a_1,\dots, a_{j-1}$, yielding a sufficiently high utility for the principal automatically provides a low utility for the agent.
	This fulfills the first inequality in \Cref{eq:opt-1-re} (and \Cref{eq:opt-1}).
	Conversely, for high-social-welfare actions $a_j,\dots, a_n$, a sufficiently low utility for the agent automatically gives a high utility for the principal.
	This fulfills the second inequality in \Cref{eq:opt-1-re} (and \Cref{eq:opt-1}).
\end{remark}

\begin{algorithm}[!htp]
	\caption{Compute an optimal $\delta$-robust contract}
	\label{alg:robust_poly}
	\begin{algorithmic}[1]
		\State $p^\star \leftarrow \texttt{null}$,\; 
		$\psi^\star \leftarrow -\infty$
		\ForAll{$(a^\star,a^\delta) \in \mathcal{A} \times \mathcal{A}$} \label{ln:inner-for}
		\ForAll{$j=1,\dots,n+1$}
		\State 
		Solve Problem~\eqref{eq:lp-nu-j-const} instantiated with $(a^\star,a^\delta)$ s.t. \Cref{eq:lp-nu-j-1,eq:lp-nu-j-2,eq:lp-nu-j-3},
		and let an optimal solution be $p'$
		
		\If{$\Psi(p') > \psi^\star$}
		\State $p^\star \leftarrow p'$ 
		\State $\psi^\star \leftarrow \Psi(p')$
		\EndIf
		
		\label{ln:inner-for-end}
		\EndFor
		\EndFor
		\State {\bf return} $p^\star$
	\end{algorithmic}
\end{algorithm}
%

\subsection{Putting All Together}

We summarize our results in this section into \Cref{alg:robust_poly} and the following main theorem.

\begin{theorem}
	\Cref{alg:robust_poly} computes an optimal $\delta$-robust contract in polynomial time.
\end{theorem}

\begin{proof}[Proof Sketch]
	The polynomial runtime of \Cref{alg:robust_poly} is obvious as it enumerates $O(n^3)$ value combinations and solves an LP for each of them.
	To see the correctness of the algorithm, note that the inner for-loop of \Cref{alg:robust_poly} effectively solves Problem~\eqref{eq:opt}, for the pair $(a^\star, a^\delta)$ enumerated in the outer for-loop.
	Now, if $a^\star$ and $a^\delta$ happen to be the agent's exact and $\delta$-best responses, respectively, under some optimal contract, then according to \Cref{lmm:Psi-opt}, the outer loop produces an optimal contract $p'$, with $\Psi(p') \ge \Psi(p)$ for all $p \in \mathbb{R}_+^m$. By comparing the $\Psi$ values, the algorithm identifies and outputs such an optimal contract.
\end{proof}


\section{Learning Robust Contracts}\label{sec:learning_problem}

The algorithm in Section~\ref{sec:computational_problem} works under the assumption that the principal has full knowledge of all the payoff-relevant information about the agent and, thus, they can compute an optimal $\delta$-robust contract.
%
%
In this section, we address the case in which the principal has no such knowledge, and, thus, they have to learn an optimal $\delta$-robust contract in an online fashion, by repeatedly interacting with the agent.
%

\subsection{Learning Interaction}\label{sec:learning_interaction}

We consider an online learning framework similar to the one studied by~\citet{Zhu2023Sample}, in which the features of agent's actions, \emph{i.e.}, costs and probabilities over outcomes, depend on an agent's \emph{type} that is sampled at each round from some (fixed) unknown probability distribution.
Before formally defining the online learning framework, we introduce some notation that is needed to deal with hidden-action principal-agent problems in which the agent can be of different types.

%

\paragraph{Hidden-action principal-agent problems with types}
We let $\Theta$ be the finite set of possible agent's {types}, while $F_{\theta,a}$ and $c_{\theta,a}$ denote the probability distribution over outcomes and the cost, respectively, of action $a \in A$, when the agent has type $\theta \in \Theta$.
The agent's type is drawn from an unknown probability distribution $\lambda \in \Delta_\Theta$, with $\lambda_\theta \in [0,1]$ being the probability of type $\theta \in \Theta$.
We extend all the notation introduced in Section~\ref{sec:preliminaries} to account for agent's types.
Specifically, given $\delta \in (0,1)$, we let $A^{\theta,\delta}(p) \subseteq \mathcal{A}$ be the set of $\delta$-best responses to a contract $p \in \mathbb{R}_+^m$ for type $\theta \in \Theta$.
Moreover, we denote by $a^{\theta,\delta}(p)$ the $\delta$-best response that is played by the agent under a $\delta$-robust contract, namely the worst one for the principal.
Formally, $a^{\theta,\delta}(p) \in \arg\min_{a \in A^\delta (p)} F_{\theta,a} \cdot (r-p)$.
Similarly, we use $A^{\theta}(p)$ and $a^{\theta}(p)$ for (exact) best responses.
Finally, the principal's expected utility when committing to a contract $p \in \Rset$ against an agent of type $\theta \in \Theta$ playing an action $a \in \mathcal{A}$ is $u^\sfP(p,a,\theta) \coloneqq F_{\theta, a} \cdot (r - p)$.

We study an online learning framework in which the principal and the agent repeatedly interact over $T > 0$ rounds.
Each round involves a repetition of the same instance of hidden-action principal-agent problem, with only the agent's type changing from round to round. 
Specifically, at each round $t \in [T]$, the principal-agent interaction is as follows:
\begin{enumerate}
	\item The agent's type $\theta^t \in \Theta$ is sampled from the distribution $\lambda$. Notice that $\theta^t$ is \emph{not} observed by the principal.
	\item The principal commits to a contract $p^t \in \mathcal{C} \coloneqq [0,1]^m$.
	\item After observing the contract $p^t$, the agent plays a $\delta$-best response $a^t\coloneqq a^{\theta^t,\delta}(p^t)$.
	Notice that the action $a^{t}$ is \emph{not} observed by the principal. 
	%
	%
	%
	\item The principal observes the outcome $\omega^t \sim {F}_{a^{t}}$ that is realized as an effect of the agent's action $a^t$.
\end{enumerate}
We remark that, in the interaction described above, the principal's contract design space is assumed to be limited to the hypercube $\mathcal{C} \coloneqq [0,1]^m$.
Restricting the contract space to a bounded set is standard in the literature on online learning in contract design (see, \emph{e.g.},~\cite{ho2015adaptive,Zhu2023Sample,chen2024bounded}) and it is motivated by the negative result proved in Theorem~1 in~\citep{bacchiocchilearning}.
Therefore, we need to introduce a formal definition of the principal's expected utility when committing to an optimal $\delta$-robust contract restricted to such a space.
Formally:
$$
%
\textnormal{OPT}(\mathcal{C}, \delta) := \max_{p \in \mathcal{C}} \sum_{\theta \in \Theta} \, \lambda_\theta u^\sfP(p, a^{\theta,\delta}(p), \theta).
$$

%
%
%
The goal of the principal is to minimize the \emph{(cumulative) pseudo-regret}, or simply \emph{regret}, which can be defined as:
\begin{align*}
	R_T(\mathcal{C},\delta) \coloneqq  T \, \textnormal{OPT}(\mathcal{C},{\delta})-   \mathbb{E} \left[\hspace{-0.2mm}\sum_{t \in [T],\theta \in \Theta} \lambda_\theta u^\sfP(p^t, a^{\theta,\delta}(p^t),\theta)  \right] ,
\end{align*}
where the expectation is over the randomness of the learning algorithm.
%
%
Our goal is to design a \emph{no-regret} algorithm for the principal, namely one achieving $R_T(\mathcal{C}, \delta) = o(T)$.



\subsection{A No-Regret Algorithm}\label{sec:no_regret}

Before introducing our no-regret learning algorithm (Algorithm~\ref{alg:no_regret}), we need to prove a key lemma about the ``continuity'' of the principal's expected utility over the hypercube.
In particular, given a $\delta$-robust contract $p \in \Rset$, we show that it is possible to build another contract $p'\in  \Rset$ that is $(\delta + \epsilon)$-robust and provides the principal with utility at most $2\sqrt{\epsilon}$ worse than the one of contract $p$. Formally:
\begin{restatable}{lemma}{epsilonconvert}\label{lem:epsilonconvert}
	Given any $\delta,\epsilon \in(0,1)$ and a contract $p \in \Rset$, the following holds for every $\theta \in \Theta$:
	$$ u^\sfP(p', a^{\theta, \delta+\epsilon}(p'),\theta) \ge u^\sfP(p,a^{\theta,\delta} (p),\theta ) - 2 \sqrt{\epsilon},$$
	where $p' \coloneqq (1-\sqrt \epsilon)p + \sqrt \epsilon r$.
\end{restatable}
%
%
%
%
%
%
%
%
We observe that the idea of shifting payments towards the principal's reward vector, as in Lemma~\ref{lem:epsilonconvert}, was first adopted by~\citet{dutting2021complexity} in non-robust settings, in order to deal with approximately incentive-compatible contracts.

Thanks to Lemma~\ref{lem:epsilonconvert}, it is possible to show that there always exists a contract $p \in \mathcal{B}_{\epsilon}$ that satisfies the following condition, where $\mathcal{B}_{\epsilon} \subseteq \mathcal{C}$ is a uniform grid of the hypercube $\mathcal{C}$, built with step size $\epsilon > 0$ (see also Algorithm~\ref{alg:no_regret}).
\[ \sum_{\theta \in \Theta} \lambda_\theta  u^{\textnormal{P}}(p,a^{\theta,\delta}(p),\theta) \ge \textnormal{OPT}(\mathcal{C},\delta) - \mathcal{O}( \sqrt{\epsilon}).\]
Consequently, by suitably choosing $\epsilon > 0$ as a function of the time horizon $T$, and by instantiating a no-regret algorithm with set of arms $\mathcal{B}_{\epsilon}$, we can upper bound the regret suffered by Algorithm~\ref{alg:no_regret} as follows.
\begin{restatable}{theorem}{NoRegretThm}\label{thm:no_regret}
	The regret suffered by Algorithm~\ref{alg:no_regret} can be upper bounded as follows:
	\begin{align*}
		R_T (\mathcal{C},\delta) \le \mathcal{\widetilde{O}} \left(   T^{1- \frac{1}{2(m+1)}}  \right).
	\end{align*}
\end{restatable}
We remark that the regret lower bound introduced by~\citet{Zhu2023Sample} also holds in our setting. 
This is because, when the parameter $\delta$ is arbitrarily small, if we consider the same instances used by~\citet{Zhu2023Sample} in their lower bound, the principal is still required to enumerate an exponential number of regions, thus suffering $\Omega(T^{1-1/(m+2)})$ regret.
This confirms that an exponential dependence on the number of outcomes is unavoidable in the regret suffered by any algorithm.
Furthermore, Theorem~\ref{thm:no_regret} shows regret guarantees similar to those obtained by~\citet{Zhu2023Sample} in the non-robust version of the problem.
Indeed, they achieve an upper bound of the order of $\widetilde{\mathcal{O}}(\sqrt{m} \cdot T^{1 - 1/(2m+1)})$ on the regret.
\begin{algorithm}[!htp]
	\caption{Regret minimizer for $\delta$-robust contracts}
	\label{alg:no_regret}
	\begin{algorithmic}[1]
		\Require $T >0 $
		\State Set $\epsilon\gets T^{-\frac{1}{m+1}}$
		\State $\mathcal{B}_{\epsilon} \hspace{-0.1mm} \gets \hspace{-0.2mm} \left\{ p \in [0,1]^m \, | \, p_\omega \in \{0,\epsilon, 2\epsilon, ..., 1 \} \, \forall \omega \in \Omega \right\} $ 
		\State Run \texttt{UCB1} with $\mathcal{B}_\epsilon $ as set of arms 
	\end{algorithmic}
\end{algorithm}

Finally, we also show that, when the parameter $\delta \in (0,1)$ is sufficiently small (with respect to the time horizon $T>0$), Algorithm~\ref{alg:no_regret} achieves sublinear regret with respect to an optimal non-robust contract within $\mathcal{C}$.
Formally, we define the value of such a contract as follows:
$$
\textnormal{OPT}(\mathcal{C}) := \max_{p \in \mathcal{C}}  \sum_{\theta \in \Theta}\lambda_\theta u^\sfP(p, a^\theta(p),\theta).
$$
Furthermore, when $\textnormal{OPT}(\mathcal{C})$ is chosen as baseline, the regret definition becomes the following:
%
\begin{align*}
	R_T(\mathcal{C}) \hspace{-0.5mm}&\coloneqq \hspace{-0.5mm} T \, \textnormal{OPT}(\mathcal{C})\hspace{-0.3mm}-  \hspace{-0.3mm}\hspace{-0.3mm} \mathbb{E}\left[\sum_{t \in [T],\theta \in \Theta} \hspace{-3.2mm}\lambda_\theta u^\sfP(p^t, a^{\theta,\delta}(p^t),\theta) \hspace{-0.3mm} \right] \hspace{-0.5mm},
\end{align*}
where the expectation is over the randomness of the learning algorithm.
We remark that the regret definition above coincides with the one introduced by~\citet{Zhu2023Sample}.
Then, we can prove that the following corollary of Theorem~\ref{thm:no_regret} holds.
\begin{restatable}{corollary}{NoRegretCor}\label{cor:no_regret}
	The regret suffered by Algorithm~\ref{alg:no_regret} can be upper bounded as follows:
	\begin{align*}
		R_T (\mathcal{C})  \le \mathcal{\widetilde{O}} \left(   T^{1- \frac{1}{2(m+1)}}  \right) +  2 \sqrt{\delta} T.
	\end{align*}
\end{restatable}
We remark that, if we set $\delta = \nicefrac{1}{T^\alpha}$ with $\alpha>0$ in Corollary~\ref{cor:no_regret}, then the regret with respect to an optimal non-robust contract within the hypercube $\mathcal{C}$ is sublinear.

\begin{remark}
	Our algorithm can be extended to deal with settings in which the agent can play \emph{any} $\delta$-best response within the set $A^{\theta^t,\delta}(p^t)$. In such a setting, the principal's utility is \emph{not} fully stochastic. However, our Algorithm~\ref{alg:no_regret} can be easily extended by instantiating an adversarial no-regret algorithm instead of \texttt{UCB1}.
\end{remark}
%

\paragraph{Comparison with~\citep{Zhu2023Sample}}
We observe that our algorithm provides some advantages compared to the state-of-the-art algorithm proposed by~\citet{Zhu2023Sample}.
First, our approach employs as set of contracts a simple discretization of the hypercube, while the approach by~\citet{Zhu2023Sample} requires determining the minimum set of contracts $\mathcal{V}_{\epsilon}(\mathcal{C})$---for some suitable $\epsilon>0$---such that, for every $p \in \mathcal{C}$, there exists a $p' \in \mathcal{V}_{\epsilon}(\mathcal{C})$ that satisfies $(r - p) \cdot p' \ge \cos(\epsilon)$.
However, the need of designing such a set of contracts makes the approach by~\citet{Zhu2023Sample} challenging to be employed in practice compared to ours.
Second, our algorithm does \emph{not} require \emph{apriori} knowledge of the principal's reward, which is instead required by the algorithm of~\citet{Zhu2023Sample}.
%



\clearpage
\bibliographystyle{plainnat}
\bibliography{references}
\newpage
\appendix
\section*{Appendix}
\RealtionNonRobust*
\begin{proof}
	We prove the two points separately.
	\begin{enumerate}
		\item Let $p'= (1-\sqrt{\delta})p + \sqrt{\delta} r $ for any contract $p \in \Rset$.
		We start by observing that the following inequality holds:
		\begin{equation*}	
			\sum_{\omega \in \Omega}F_{ a(p),\omega} p_\omega - c_{ a(p)} \ge \sum_{\omega \in \Omega}F_{ a^\delta(p'),\omega} p_\omega - c_{ a^\delta(p')},
		\end{equation*}
		because $a(p) \in \mathcal{A}(p)$. Furthermore, we have:
		\begin{equation*}	
			\sum_{\omega \in \Omega}F_{ a^\delta(p'),\omega} p_\omega' - c_{ a^\delta(p')} > \max_{a' \in \mathcal{A}} \sum_{\omega \in \Omega}F_{ a',\omega} p_\omega - c_{ a'} -\delta \ge \sum_{\omega \in \Omega}F_{ a(p),\omega} p_\omega' - c_{ a(p)} - \delta,
		\end{equation*}
		thanks to the definition of $\A^\delta(p')$.
		Then, by employing the two above inequalities and the definition of $p'$, we have:
		\begin{align*}	
			\delta & \ge \sum_{\omega \in \Omega}F_{ a(p),\omega} p_\omega' - c_{ a(p)} - \left(\sum_{\omega \in \Omega}F_{ a^\delta(p'),\omega} p_\omega' - c_{ a^\delta(p')} \right)\\
			& =  \sum_{\omega \in \Omega}F_{ a(p),\omega} p_\omega - c_{ a(p)} - \left( \sum_{\omega \in \Omega}F_{ a^\delta(p'),\omega} p_\omega - c_{ a^\delta(p')} \right)  + \sqrt \delta \left( \sum_{\omega \in \Omega}\left(F_{ a(p),\omega} - F_{ a^\delta(p),\omega}\right) (r_\omega - p_\omega) \right) \\ 
			& \ge \sqrt \delta \left( \sum_{\omega \in \Omega}\left(F_{ a(p),\omega} - F_{ a^\delta(p),\omega}\right) (r_\omega - p_\omega) \right). 
		\end{align*}
		Thus, by rearranging the latter inequality we get:
		\begin{equation}\label{eq:sqrtdelta}
			\sqrt \delta  \ge \left( \sum_{\omega \in \Omega}\left(F_{ a(p),\omega} - F_{ a^\delta(p),\omega}\right) (r_\omega - p_\omega) \right). 
		\end{equation}
		Finally, we can show that:
		\begin{align*}	
			u^\sfP(p,a(p))  - u^\sfP(p',a^{\delta}(p'))  &= \sum_{\omega \in \Omega}F_{ a(p),\omega} (r_\omega - p_\omega ) - \left(\sum_{\omega \in \Omega}F_{ a^\delta(p'),\omega}(r_\omega - p_\omega' )  \right)\\
			&= \sum_{\omega \in \Omega}F_{ a(p),\omega} (r_\omega - p_\omega ) - \left(1 -\sqrt{\delta}\right) \left(\sum_{\omega \in \Omega}F_{ a^\delta(p'),\omega}(r_\omega - p_\omega)  \right)\\
			&\le \left(1 -\sqrt{\delta}\right) \left(\sum_{\omega \in \Omega} (F_{ a(p),\omega} - F_{ a^\delta(p'),\omega})(r_\omega - p_\omega)  \right) + \sqrt{\delta} \\
			&\le 2 \sqrt{\delta} - \delta,
		\end{align*}
		where the second equality holds thanks to the definition of $p'$, the first inequality because $u^\sfP(p,a(p)) \le 1$ for each $p \in \mathbb{R}^{m}_{+}$ and the second inequality because of \Cref{eq:sqrtdelta}.
		Finally, let $p \in \mathbb{R}$ be an optimal (non-robust) contract, then we have:
		\begin{align*}	
			\textnormal{OPT} - 2 \sqrt{\delta} + \delta \le  u^\sfP(p',a^{\delta}(p')) \le \textnormal{OPT}(\delta) ,
		\end{align*}
		concluding the first part of the proof. 
		\item We split the proof into two parts.
		\begin{enumerate}
			\item if $\textnormal{OPT}(\delta) = 0$. Then, we trivially have :
			\begin{equation*}
				0=\textnormal{OPT}(\delta)\le \max\left(0, \, \max_{a \in \mathcal{A}} \sum_{\omega \in \Omega}F_{a,\omega}r_\omega - c_a - \delta \right) = \max\left(0, \, \textnormal{SW} - \delta \right)
			\end{equation*}
			\item if $\textnormal{OPT}(\delta) > 0$.
			We define $p^\star \in \Rset$ as an optimal $\delta$-robust contract.
			Then, we notice that $a_{1} \not \in \mathcal{A}^{\delta}(p^\star)$, where $a_1$ is the opt-out action.
			Indeed, if $a_{1} \in \mathcal{A}^{\delta}(p^\star)$, then the agent may select the action $a_1$ as a $\delta$-robust best-response, which provides zero or negative utility to the principal and contradicts the fact that $\textnormal{OPT}(\delta) > 0$.
			Therefore, in an optimal robust contract $p^\star$, we must have:
			\begin{align*}
				\sum_{\omega \in \Omega} F_{a',\omega} p_\omega^\star - c_{a'}  &\ge  \sum_{\omega \in \Omega}  F_{a_1,\omega} p_\omega^\star - c_{a_1} + \delta\ge  \sum_{\omega \in \Omega}  F_{a_1,\omega} p_\omega^\star + \delta \ge \delta,
			\end{align*}
			for some $a' \in \mathcal{A}^{\delta}(p^\star)$.
			Thus, we have:
			\begin{align*}
				\textnormal{OPT}(\delta)
				& = \sum_{\omega \in \Omega} F_{a^\delta(p),\omega} (r_\omega -p_\omega^\star ) \\
				& \le \sum_{\omega \in \Omega} F_{a',\omega} (r_\omega -p_\omega^\star ) \\
				& \le \sum_{\omega \in \Omega} F_{a',\omega} r_\omega - c_{a'}  - \delta \\
				& \le  \max_{a \in \mathcal{A}} \sum_{\omega \in \Omega}F_{a,\omega}r_\omega - c_a - \delta \\
				& \le \max\left(0, \, \max_{a \in \mathcal{A}} \sum_{\omega \in \Omega}F_{a,\omega}r_\omega - c_a - \delta \right) = \max\left(0, \, \textnormal{SW} - \delta \right).
			\end{align*}
			The first inequality above follows from the fact that $a' \in \mathcal{A}^{\delta}(p^\star)$, while the second inequality holds due to the previous observation.
		\end{enumerate}
	\end{enumerate}
	The two above points conclude the proof.
\end{proof}
\RealtionNonRobustTwo*
\begin{proof}
	We prove the two points separately.
	\begin{enumerate}
		\item We consider an instance parametrized by $\delta > 0$, with $|\Omega| = 2$ and $|\mathcal{A}| = 2n + 1$ for some $n \in \mathbb{N}$, to be defined below. We let:
		$$\kappa  = \min \Bigg\{i\in \mathbb{N}_{>0} \, | \, \sqrt{\delta } < \frac{i- 1}{i} \Bigg\}.$$
		Furthermore, we introduce the following decreasing sequence:
		$$\gamma_i = \begin{cases}
			\vspace{1mm}
			\frac{i}{i-1} \,\,\ & i = \kappa, \dots, n \\
			\frac{2n+1-i}{2n+2-i} \,\,\ &i = n+2, \dots, 2n,\\
		\end{cases}$$
		with $\gamma_{n+1}=1$ and $\gamma_{2n+1}=0$, where $n \in \mathbb{N}$ is such that $n > \kappa$.
		The distributions over the set of outcomes for the different actions and their corresponding costs are given by:
		$$ \begin{cases}
			F_{a_i,\omega_1}=0 \,\,& c_{a_i}=0, \,\, i=1, \dots, \kappa-1.   \\
			F_{a_i,\omega_1}=1-\gamma_i \sqrt{\delta}\,\,& c_{a_i}=0, \,\, i=\kappa, \dots, n . \\
			F_{a_{2n+1,\omega_1}}=1 \,\,& c_{a_{2n+1}}=0.
		\end{cases}$$
		and the principal's reward is $r=(1,0)$.
		Notice that the distribution over outcomes of the different actions are always well defined because of the definition of $\kappa>0$.
		Furthermore, since all the agent's actions $a_i$ with $i < k$ are coincident, we assume for the sake of presentation that the agent always selects $a_1$.
		It is easy to verify that the value of an optimal (non-robust) contract is $\textnormal{OPT}=1$ since the agent breaks ties optimistically and $c_{2n+1}=0$.	
		
		With a similar argument to the one proposed by~\cite{dutting2024algorithmic} in Proposition 3.9, an optimal $\delta$-robust contract is such that $p^\star = (0, \alpha)$ for some $\alpha \in \mathbb{R}_{+}$.
		Consequently, in the rest of the proof, we focus on determining the maximum utility achievable in a linear, $\delta$-robust contract. 
		
		%

		For every $a_i \in \mathcal{A}$, with $\kappa<i \le 2n+1$, if $\alpha \in \mathbb{R}_{+}$ is such that $a^\delta(\alpha r)= a_i$, then the two following conditions hold.
		\begin{itemize}
			\item All the actions $a_j \not \in \mathcal{A}^\delta(\alpha r)$ for each $j < i$.
			This is because $u^\sfP(\alpha, a_j) \le u^\sfP(\alpha, a_i)$ for each $\alpha \in \mathbb{R}_{+}$ since $\{\gamma_i\}_{i \in [2n+1]}$ is a decreasing sequence.
			Thus, the latter observation implies that:
			\begin{align*}
				u^\sfA(\alpha, a_{j}) 
				\le u^\sfA(\alpha, a_{i-1}) & = \alpha R_{a_{i-1}} - c_{a_{i-1}}\\
				& \le \max_{i \in [2n+1]} u^\sfA(\alpha, a_{i}) -\delta\\
				& = u^\sfA(\alpha, a_{{2n+1}}) -\delta  \\
				&= \alpha R_{a_{2n+1}} -c_{a_{{2n+1}}} - \delta \\
				& = \alpha -\delta.
			\end{align*}
			Therefore, we have that: 
			\begin{align*}
				\alpha R_{i-1}\le \alpha - \delta, 
			\end{align*}
			and, thus, $\alpha \ge {\sqrt{\delta}} / {\gamma_{i-1}}$. 
			
			\item The action $a_i \in \A^\delta(\alpha r)$. Thus, 
			$$ \alpha R_{a_i} - c_{a_i} > \max_{i \in [2n+1]} u^\sfA(\alpha, a_{i}) -\delta = \alpha -\delta,
			$$
			which implies that $\alpha < {\sqrt{\delta}} / {\gamma_{i}}$.
		\end{itemize}
		The two observation above shows that $a_i$, with $i > \kappa$, is a $\delta$-best response for all the values of $\alpha$ such that:
		$${\sqrt{\delta}} / {\gamma_{i-1}} \le \alpha < {\sqrt{\delta}} / {\gamma_{i}}.$$
		Thus, the smallest value of $\alpha \in \mathbb{R}_{+}$ such that $a^\delta(\alpha r) = a_i$ is $\alpha_i \coloneqq {\sqrt{\delta}} / {\gamma_{i-1}}$, when $i > \kappa$.
		With the same argument above, it is possible to show that $a_1=a^\delta(\alpha r)$ for all $\alpha \in [\alpha_1,\alpha_\kappa)$ and $a_\kappa=a^\delta(\alpha r)$ for all $\alpha \in [\alpha_\kappa, \alpha_{\kappa+1})$, with $\alpha_1=0$ and $\alpha_\kappa=\delta$.
		
		We also observe that the principal's utility in each $\alpha_i$, with $i \in \{\kappa+1, \dots, 2n+1\}$, is such that:
		\begin{align*}
			u^\sfP(\alpha_i, a_{i}) =  (1-\alpha_i) R_{a_i} & = \left(1 - \frac{\sqrt{\delta}}{\gamma_{i-1}}\right) \left(1-\gamma_i \sqrt{\delta} \right)\\
			& = \left(1 - \frac{\sqrt{\delta}}{\gamma_{i-1}}\right) \left(1-\gamma_i \sqrt{\delta} \right)\\
			& = \left( 1- \left(\gamma_i  + \frac{1}{\gamma_{i-1}}\right)\right) \sqrt{\delta} +  \frac{\gamma_{i}}{\gamma_{i-1}} \delta.
		\end{align*}
		Therefore, using the latter formula, we have that the following holds.
		\begin{itemize}
			\item If $i=\kappa +1,\dots,n$, we have:
			$$u(\alpha_i) \le 1- 2\sqrt{\delta} + \frac{\gamma_i}{\gamma_{i-1}}\delta \le  1 - 2\sqrt{\delta} + \delta, $$
			since $\{\gamma_i\}_{i \ge \kappa}$ is  a decreasing sequence.
			
			\item If $i=n+1,n+2$, we have:
			$$u(\alpha_{i}) \le 1- \left(1 + \frac{n-1}{n} \right)\sqrt{\delta} +  \frac{n-1}{n} \delta \le 1- 2\sqrt{\delta} + \delta + \frac{\sqrt{\delta}}{n}.$$
			
			\item If $i=n+3, \dots, 2n+1$, we have:
			$$u(\alpha_{i}) \le 1- 2\sqrt{\delta} + \frac{\gamma_i}{\gamma_{i-1}}\delta \le 1- 2\sqrt{\delta} + \delta ,$$
			since $\{\gamma_i\}_{i \ge \kappa}$ is  a decreasing sequence.
		\end{itemize}
		
		We also observe that the principal's utility in $\alpha_{\kappa}$ is such that:
		\begin{align*}
			u(\alpha_{\kappa}) = (1-\delta)(1- \gamma_{\kappa} \sqrt{\delta} ).
		\end{align*}
		since $\gamma_k>1$ and $\delta>0$.
		Thanks to the definition of $\kappa$, with a simple calculation, it possible to show that:
		$$\kappa = \begin{cases}
			\vspace{1mm}
			\left\lceil\frac{1}{1-\sqrt{\delta}}\right\rceil \,\,\ & \textnormal{if} \,\,\, \frac{1}{1-\sqrt{\delta}} \not \in \mathbb{N}\\
			\frac{1}{1-\sqrt{\delta}} +1 \,\,\ & \textnormal{if} \,\,\, \frac{1}{1-\sqrt{\delta}} \in \mathbb{N}.\\
		\end{cases}$$
		Thus, when ${1}/{(1-\sqrt{\delta})}\not \in \mathbb{N}$, we have:
		\begin{align*}
			u(\alpha_{\kappa}) &= (1-\delta)(1- \gamma_{\kappa} \sqrt{\delta} )\\
			&= (1-\delta)\left(1- 
			\frac{\left\lceil\frac{1}{1-\sqrt{\delta}}\right\rceil }{ \left\lceil\frac{1}{1-\sqrt{\delta}}\right\rceil -1} \sqrt{\delta}\right)\\
			&\le (1-\delta)\left(1-2\sqrt{\delta}+\delta\right).
		\end{align*}
		Similarly, it is possible to show that the same relation holds even when ${1}/{(1-\sqrt{\delta})} \in \mathbb{N}$.
		Finally, by putting all together, we have that:
		\begin{align*}
			\textnormal{OPT}(\delta) -\textnormal{OPT}_{\textnormal{LB}}(\delta) 
			& \le \max_{i \in [2n+1]} u(\alpha_i) - (1 - 2 \sqrt{\delta} + \delta )\\
			&= u(\alpha_{n+1}) - (1 - 2 \sqrt{\delta} + \delta ) \le \frac{\sqrt{\delta}}{n},
		\end{align*}
		concluding the first part of the proof.
		\item We consider an instance with $|\Omega|=|\mathcal{A}|=2$ and $r=(0,1)$.
		The distributions over the set of outcomes for the different actions and their corresponding costs are given by:
		$$ \begin{cases}
			F_{a_1}=(1,0) \,\, &c_{a_1}=0\\
			F_{a_2}=(0,1) \,\, &c_{a_2}=0.
		\end{cases}$$
		%
		%
		It is easy to verify that $\textnormal{SW}=R_{a_2}-c_{a_2}=1$. 
		Furthermore, with a similar argument to the one proposed by~\cite{dutting2024algorithmic} in Proposition 3.9, an optimal $\delta$-robust contract is such that $p^\star = (0, \alpha)$ for some $\alpha \in \mathbb{R}_{+}$.

		We show that $\textnormal{OPT}(\delta) = 1-\delta.$ 
		Let $a_1 \in \mathcal{A}$ be the opt-out action.
		We observe that $a_1 \not \in {A}^{\delta}(\alpha r)$ for all the $\alpha \in \mathbb{R}_{+}$ that satisfy:
		\begin{equation*}
			\alpha R_{a_2} - c_{a_2} = \alpha \ge \alpha R_{a_1} - c_{a_1} + \delta = \delta.
		\end{equation*}
		Thus, $a_1 \not \in {A}^{\delta}(\alpha r)$ for all $\alpha \ge \delta$.
		Therefore, the smallest $\alpha \in \mathbb{R}_{+}$ such that the agent selects the action $a_2$ coincides with $\delta > 0$.
		Thus, the largest utility the principal can achieve satisfies $\textnormal{OPT}(\delta) = 1 - \delta = \textnormal{SW} - \delta$.
	\end{enumerate}
	The two points above conclude the proof.
\end{proof}
\RealtionNonRobustThree*
\begin{proof}
	We prove the claims separately.
	\begin{itemize}
		\item We first show that $\textnormal{OPT}(\delta)$ is a continuous function. 
		Let $p^\star \in \Rset$  be a $\delta$-robust and $p' = (1-\sqrt{\epsilon})p^\star + \sqrt{\epsilon} r$. Then, by Lemma~\ref{lem:epsilonconvert}, we have:
		\begin{equation*}
			\textnormal{OPT}(\delta+\epsilon) \ge u^{\textnormal{P}}(p',a^{\delta+\epsilon}(p')) \ge u^{\textnormal{P}}(p^\star,a^{\delta} (p^\star) ) - 2 \sqrt{\epsilon} + \epsilon = \textnormal{OPT}(\delta) - 2 \sqrt{\epsilon} + \epsilon.
		\end{equation*}
		Thus, rewriting the latter inequality and taking the limit, we have that:
		\begin{equation*}
			0= \lim_{\epsilon \to 0} (2 \sqrt{\epsilon} - \epsilon) \ge \lim_{\epsilon \to 0} \textnormal{OPT}(\delta) - \textnormal{OPT}(\delta+\epsilon).
		\end{equation*} 
		Let $p^\star \in \Rset$  be a $\delta'$-robust and $p' = (1-\sqrt{\epsilon})p^\star + \sqrt{\epsilon} r$ with $\delta'=\delta-\epsilon$. Then, by Lemma~\ref{lem:epsilonconvert}, we have: 
		\begin{equation*}
			\textnormal{OPT}(\delta)= \textnormal{OPT}(\delta'+\epsilon) \ge u^{\textnormal{P}}(p',a^{\delta'+\epsilon}(p')) \ge u^{\textnormal{P}}(p^\star,a^{\delta'} (p^\star) ) - 2 \sqrt{\epsilon} + \epsilon = \textnormal{OPT}(\delta-\epsilon) - 2 \sqrt{\epsilon} + \epsilon.
		\end{equation*}
		Thus, by taking the limit, we have that:
		\begin{equation*}
			0= \lim_{\epsilon \to 0} (2 \sqrt{\epsilon} - \epsilon) \ge \lim_{\epsilon \to 0} \textnormal{OPT}(\delta-\epsilon) - \textnormal{OPT}(\delta).
		\end{equation*} 
		Since the latter considerations hold for every $\delta \in (0,1)$, we can easily show that $\textnormal{OPT}(\delta)$ is continuous in $\delta \in (0,1)$.
		\item We now prove that $\textnormal{OPT}(\delta)$ is non-increasing. 
		Let $\delta,\delta' \in (0,1)$ such that $\delta'<\delta$. Furthermore, let $p^\star \in \Rset$ be an optimal $\delta$-robust contract. Therefore, we have:
		\begin{equation*}
			\textnormal{OPT}(\delta)=u(a^\delta(p^\star), p^\star ) \le u(a^{\delta'}(p^\star), p^\star ) \le \textnormal{OPT}(\delta'),
		\end{equation*}
		since, for each $p \in \Rset$, we have ${A}^{\delta'}(p) \subseteq  {A}^{\delta}(p)$.
		\item By Proposition~\ref{pro:properties}, we have:
		\begin{equation*}
			\lim_{\delta \to 0^+}  \textnormal{OPT}(\delta) \ge \lim_{\delta \to 0^+}\textnormal{OPT} - 2 \sqrt{\delta} + \delta =\textnormal{OPT},
		\end{equation*}
		and,
		\begin{equation*}
			\lim_{\delta \to 1^-}  \textnormal{OPT}(\delta) \le \lim_{\delta \to 1^-}1-\delta =0.
		\end{equation*}
	\end{itemize}
	The two points above conclude the proof.
\end{proof}
\epsilonconvert*
\begin{proof}
	For the sake of the presentation, we avoid the dependence on the type $\theta \in \Theta$.
	We split the proof in two cases: 
	\begin{enumerate}
		\item If $a^{\delta+\epsilon} (p') \not \in {A}^{\delta}(p)$, then $a^{\delta+\epsilon} (p')$ is not a $\delta$-best-response in $p$. Therefore, we have that:
		\begin{equation*}	
			\sum_{\omega \in \Omega} F_{ a^{\delta}(p),\omega} p_\omega - c_{ a^{\delta}(p)} \ge \sum_{\omega \in \Omega}F_{ a^{\delta+\epsilon}(p'),\omega} p_\omega - c_{ a^{\delta+\epsilon}(p')} + {\delta },
		\end{equation*}
		and, 
		\begin{equation*}	
			\sum_{\omega \in \Omega}F_{ a^{\delta+\epsilon}(p'),\omega} p_\omega' - c_{ a^{\delta+\epsilon}(p')}  \ge \sum_{\omega \in \Omega} F_{ a^{\delta}(p),\omega} p_\omega' - c_{ a^{\delta}(p)} - {\delta - \epsilon }.
		\end{equation*}
		%
		%
		Then, by taking the summation of the above quantities and employing the definition of $p'$  we get:
		\begin{equation*}	
			\sqrt \epsilon \ge \sum_{\omega \in \Omega} (F_{ a^{\delta}(p),\omega} - F_{ a^{\delta + \epsilon}(p'),\omega}) (r_\omega-p_\omega).
		\end{equation*}
		Therefore, we can prove that the following holds:
		\begin{align*}	
			u^\sfP(p,a^{\delta}(p)) & - u^\sfP(p',a^{\delta+ \epsilon}(p'))  = \sum_{\omega \in \Omega}F_{ a^{\delta}(p),\omega} (r_\omega - p_\omega ) - \left(\sum_{\omega \in \Omega}F_{ a^{\delta+\epsilon}(p'),\omega}(r_\omega - p_\omega' )  \right)\\
			&= \sum_{\omega \in \Omega}F_{ a^{\delta}(p),\omega} (r_\omega - p_\omega ) - \left(1 -\sqrt{\epsilon}\right) \left(\sum_{\omega \in \Omega}F_{ a^{\delta+\epsilon}(p'),\omega}(r_\omega - p_\omega)  \right)\\
			&\le \left(1 -\sqrt{\epsilon}\right) \left(\sum_{\omega \in \Omega} (F_{ a^{\delta}(p),\omega} - F_{ a^{\delta+\epsilon}(p'),\omega})(r_\omega - p_\omega)  \right) + \sqrt{\epsilon} \\
			&\le 2 \sqrt{\epsilon} - \epsilon.
		\end{align*}
		\item if $a^{\delta+\epsilon} (p') \in {A}^{\delta}(p)$, then one of the following hold.
		\begin{enumerate}
			\item if $a^\delta(p) \in {A}^{\delta+\epsilon}(p')$, then either $a^\delta(p) \equiv a^{\delta+\epsilon}(p')$ or $a^\delta(p)$  provides the same principal's utility as $a^{\delta+\epsilon} (p')$ in $p'$. Indeed, suppose by contradiction that the following holds:
			\begin{equation*}	
				\sum_{\omega \in \Omega} F_{ a^{\delta+\epsilon}(p'),\omega} (r_\omega-p_\omega') < \sum_{\omega \in \Omega} F_{ a^\delta(p),\omega} (r_\omega-p_\omega').
			\end{equation*}
			Then, we have:
			\begin{align*}	
				\sum_{\omega \in \Omega} F_{ a^{\delta+\epsilon}(p'),\omega} (r_\omega-p_\omega') 
				& = (1-\sqrt{\epsilon}) \sum_{\omega \in \Omega} F_{ a^{\delta+\epsilon}(p'),\omega} (r_\omega-p_\omega) \\
				& \ge (1-\sqrt{\epsilon}) \sum_{\omega \in \Omega} F_{ a^\delta(p),\omega} (r_\omega-p_\omega) \\
				& =  \sum_{\omega \in \Omega} F_{ a^\delta(p),\omega} (r_\omega-p_\omega').
			\end{align*}
			where the inequality is a consequence of the fact that both $a^\delta(p)$ and $a^{\delta+\epsilon}(p')$ belongs to ${A}^{\delta} (p)$. This, reaches a contradiction with the initial assumption. Thus, we have that:
			\begin{align*}	
				u^\sfP (p',a^{\delta+\epsilon}(p')) & = \sum_{\omega \in \Omega} F_{ a^{\delta+\epsilon}(p'),\omega} (r_\omega-p_\omega')\\ 
				&= \sum_{\omega \in \Omega} F_{ a^\delta(p),\omega} (r_\omega-p_\omega')\\
				& = (1-\sqrt{\epsilon}) \sum_{\omega \in \Omega} F_{ a^{\delta}(p),\omega} (r_\omega-p_\omega)\\
				& \ge \sum_{\omega \in \Omega} F_{ a^{\delta}(p),\omega} (r_\omega-p_\omega) -\sqrt{\epsilon}\\
				& \ge u^\sfP (p, a^\delta(p)) -\sqrt{\epsilon}.
			\end{align*}

			\item if $a^\delta(p) \not\in {A}^{\delta+\epsilon}(p')$, then the following holds:
			\begin{align*}	
				\sum_{\omega \in \Omega} F_{ a^{\delta}(p'),\omega} p_\omega - c_{ a^{\delta}(p)} & \ge \sum_{\omega \in \Omega} F_{ a^{\delta+\epsilon}(p'),\omega} p_\omega - c_{ a^{\delta+\epsilon}(p')} - {\delta },
			\end{align*}
			and, 
			\begin{align*}	
				\sum_{\omega \in \Omega}F_{ a^{\delta+\epsilon}(p'),\omega} p_\omega' - c_{ a^{\delta+\epsilon}(p')}  & \ge \sum_{\omega \in \Omega} F_{ a^{\delta}(p),\omega} p_\omega' - c_{ a^{\delta}(p)} + \delta +\epsilon.
			\end{align*}
			Then, by taking the summation of the above quantities and employing the definition of $p'$  we get:
			\begin{equation*}	
				0 \ge \sum_{\omega \in \Omega} (F_{ a^{\delta}(p),\omega} - F_{ a^{\delta+\epsilon}(p'),\omega}) (r_\omega-p_\omega).
			\end{equation*}
			Then, using the same argument employed at point (a) we can show that:
			\begin{align*}	
				u^\sfP (p,a^{\delta}(p))  - u^\sfP (p',a^{\delta+\epsilon}(p')) \le  \sqrt{\epsilon},
			\end{align*}
			concluding the proof.
		\end{enumerate}
	\end{enumerate}
\end{proof}

\NoRegretThm*
\begin{proof}
	In the following, we let:
	\begin{enumerate}
		\item $p^\star$ be such that $\sum_{\theta \in \Theta}\lambda_\theta u^\sfP(p^\star,a^{\theta,\delta}(p^\star),\theta)=\textnormal{OPT}(\mathcal{C},\delta)$. 
		\item $p' \coloneqq (1-\sqrt{2\epsilon})p^\star + \sqrt{2\epsilon} r$. Notice that $p' \in [0,1]^m $ since both $p,r \in \mathcal{C}$ and $\mathcal{C}$ is convex.
		\item $\overline p \in \mathcal{B}_\epsilon$ be such that $\| \overline p-p'\|_{\infty} \le \epsilon$. Notice that there always exists at least once contract $\bar p$ satisfying the latter condition because of the definition of $\mathcal{B}_\epsilon$ and the fact that $p' \in [0,1]^m $.
		
	\end{enumerate}
	We start by observing that, for each $\theta \in \Theta$, it holds $a^{\theta,\delta}(\overline{p}) \in {A}^{\theta, \delta + 2 \epsilon}(p')$. Indeed, we have:
	\begin{align*}	
		\sum_{\omega \in \Omega} F_{\theta, a^{\theta,\delta}(\overline{p}),\omega} p_\omega' - c_{ \theta,a^{\theta,\delta}(\overline{p})}
		& \ge \sum_{\omega \in \Omega} F_{ \theta,a^{\theta,\delta}(\overline{p}),\omega} \overline{p}_\omega - c_{ \theta,a^{\theta,\delta}(\overline{p})} - \epsilon \\
		& > \sum_{\omega \in \Omega} F_{ \theta,a,\omega} \overline{p}_\omega - c_{ \theta,a} - \epsilon - \delta \\
		& \ge \sum_{\omega \in \Omega} F_{ \theta,a,\omega} p_\omega' - c_{\theta, a} - 2\epsilon - \delta,
	\end{align*}
	for every actions $a \in \A$. The above inequalities follows since:
	$$\sum_{\omega \in \Omega}F_{\theta,a, \omega} (\overline{p}_\omega-p_\omega') \le \|F_{\theta,a}\|_1 \|\overline{p}-p'\|_{\infty} \le \epsilon,$$ for every actions $a \in \A$ and $\theta \in \Theta$.

	Therefore, we can prove that the following holds:
	\begin{align}
		\sum_{\theta \in \Theta } \lambda_\theta u^\sfP(\overline{p}, a^{\theta,\delta}(\overline{p}),\theta) &\ge \sum_{\theta \in \Theta } \lambda_\theta  u^\sfP(p', a^{\theta,\delta}(\overline{p}),\theta) - \epsilon \nonumber \\
		&\ge \sum_{\theta \in \Theta } \lambda_\theta  u^\sfP(p', a^{\theta,\delta+2\epsilon}(p'),\theta) - \epsilon \nonumber \\
		&\ge \sum_{\theta \in \Theta } \lambda_\theta  u^\sfP(p^\star, a^{\theta,\delta}(p^\star),\theta) - 2\sqrt{2\epsilon} \nonumber \\
		&= \textnormal{OPT}(\mathcal{C},\delta) - 2\sqrt{2\epsilon},\label{eq:eps_optimal}
	\end{align}
	where the first inequality above holds because $\sum_{\omega \in \Omega}F_{\theta, a, \omega} (\overline{p}_\omega-p_\omega') \le \|F_{\theta,a}\|_1 \|\overline{p}-p'\|_{\infty} \le \epsilon$ for every actions $a \in \A$ and type $\theta \in\ \Theta$. The second inequality holds because of the definition of $\delta$-best response, while the third inequality follows from Lemma~\ref{lem:epsilonconvert}.

	At this point, we can decompose the cumulative regret as follows:
	\begin{align}
		&R_T(\mathcal{C},\delta)   =   T \cdot \textnormal{OPT}(\mathcal{C},\delta) -  \mathbb{E}\Big[\sum_{t=1}^{T} \sum_{\theta \in \Theta} \lambda_\theta u^\sfP(p^t, a^{\theta,\delta}(p^t),\theta)  \Big] \nonumber\\
		& \le \underbrace{T \cdot \textnormal{OPT}(\mathcal{C},\delta)  - T \cdot \max_{p \in \mathcal{B}_\epsilon} \sum_{\theta \in \Theta} \hspace{-1mm} \lambda_\theta u^\sfP(p,a^{\theta,\delta}(p),\theta) }_{(a)} \\
		&+ \underbrace{ T\cdot \max_{p \in \mathcal{B}_\epsilon}  \sum_{\theta \in \Theta} \hspace{-1mm} \lambda_\theta u^\sfP(p,a^{\theta,\delta}(p),\theta) -  \mathbb{E}\Big[ \sum_{t \in [T],\theta \in \Theta} \hspace{-1mm} \lambda_\theta u^\sfP(p^t, a^{\theta,\delta}(p^t),\theta) \Big]}_{(b)}. \label{eq:regret_decomp}
	\end{align}
	We focus on bounding term (a) in \cref{eq:regret_decomp}.
	Let $\bar p \in [0,1]^m$ be the contract defined at point 3. Then, we have: 
	\begin{align*}
		T \cdot \textnormal{OPT}(\mathcal{C},\delta)- T \max_{p \in \mathcal{B}_\epsilon}   \sum_{\theta \in \Theta} \lambda_\theta  u^\sfP(p^t, a^{\theta,\delta}(p^t),\theta) 
		& \le  T \cdot \textnormal{OPT}(\mathcal{C},\delta)- T \sum_{\theta \in \Theta} \lambda_\theta  u^\sfP(\bar p, a^{\theta,\delta}(\bar p),\theta)  \\
		& \le  2 \sqrt{2\epsilon} \, T,
	\end{align*}
	where the last inequality holds because of \Cref{eq:eps_optimal}.
	We focus on bounding term (b) in \cref{eq:regret_decomp}. By employing the same analysis to prove the regret bound in \texttt{UCB1} we have:
	\begin{equation}
		T \cdot \max_{p \in \mathcal{B}_\epsilon}  \sum_{\theta \in \Theta} \lambda_\theta u^\sfP(p,a^{\theta,\delta}(p),\theta) -  \mathbb{E}\Big[ \sum_{t \in [T]} \sum_{\theta \in \Theta} \lambda_\theta  u^\sfP(p^t, a^{\theta,\delta}(p^t),\theta) \Big] \le \mathcal{O}\left(\sqrt{ T |\mathcal{B}_\epsilon| \log(T ) }\right).
	\end{equation} 
	Finally, observing that $|\mathcal{B}_\epsilon| \le \mathcal{O}{((\nicefrac{1}{\epsilon})}^{m})$, by putting (a) and (b) in \cref{eq:regret_decomp} together and setting $\epsilon= T^{-\frac{1}{m+1}}$, we have:
	\begin{align*}
		R_T (\mathcal{C},\delta) \le \mathcal{\widetilde{O}} \left(   T^{1- \frac{1}{2(m+1)}}  \right),
	\end{align*}
	concluding the proof.
\end{proof}

\NoRegretCor*
\begin{proof}
	Let $p^\star$ be an optimal non-robust contract inside $\mathcal{C}$. Then we have:
	\begin{align*}	
		R_T (\mathcal{C}) & =   T \cdot \textnormal{OPT}(\mathcal{C}) -  \mathbb{E}\Big[\sum_{t=1}^{T} \sum_{\theta \in \Theta} \lambda_\theta u^\sfP(p^t, a^{\theta,\delta}(p^t),\theta)  \Big] \\
		&= T \cdot \left( \textnormal{OPT}(\mathcal{C}) -  \textnormal{OPT}(\mathcal{C}, \delta)\right) + T \cdot \textnormal{OPT}(\mathcal{C}, \delta) -  \mathbb{E}\Big[\sum_{t=1}^{T} \sum_{\theta \in \Theta} \lambda_\theta u^\sfP(p^t, a^{\theta,\delta}(p^t),\theta)  \Big] \\
		&\le 2 \sqrt{\delta} T + T \cdot \textnormal{OPT}(\mathcal{C},\delta)  -  \mathbb{E}\Big[\sum_{t=1}^{T} \sum_{\theta \in \Theta} \lambda_\theta u^\sfP(p^t, a^{\theta,\delta }(p^t),\theta)  \Big]  \\
		& \le 2 \sqrt{\delta} T + \mathcal{\widetilde{O}} \left(   T^{1- \frac{1}{2(m+1)}}  \right) ,
	\end{align*}
	where the first inequality holds by employing the same argument needed to prove Proposition~\ref{pro:properties}, and the second inequality holds thanks to Theorem~\ref{thm:no_regret}, concluding the proof.
\end{proof}


\end{document}